\documentclass[10pt,conference,twocolumn]{IEEEtran}
\usepackage{graphicx,amssymb,amsmath}
\usepackage{subfigure}
\usepackage[noadjust]{cite}
\newtheorem{lemma}{Lemma}%
\usepackage{algorithm,algorithmic}
\DeclareMathOperator*{\argmax}{argmax}

\title{Delay-Aware Dynamic Resource Management for High-Speed Railway Wireless Communications}

\author{Shengfeng Xu, Gang Zhu, Chao Shen, Shichao Li, Zhangdui Zhong\\State Key Laboratory of Rail Traffic Control and Safety\\Beijing Jiaotong University, Beijing, P. R. China
\\Email:$\{$11120177, gzhu, shenchao, 14111022, zhdzhong$\}$@bjtu.edu.cn}

\begin{document}

\maketitle

\begin{abstract}
In this paper, we investigate the delay-aware dynamic resource management problem for multi-service transmission in high-speed railway wireless communications, with a focus on resource allocation among the services and power control along the time. By taking account of average delay requirements and power constraints, the considered problem is formulated into a stochastic optimization problem, rather than pursuing the traditional convex optimization means. Inspired by Lyapunov optimization theory, the intractable stochastic optimization problem is transformed into a tractable deterministic optimization problem, which is a mixed-integer resource management problem. By exploiting the specific problem structure, the mixed-integer resource management problem is equivalently transformed into a single variable problem, which can be effectively solved by the golden section search method with guaranteed global optimality.
Finally, we propose a dynamic resource management algorithm to solve the original stochastic optimization problem. Simulation results show the advantage of the proposed dynamic algorithm and reveal that there exists a fundamental tradeoff between delay requirements and power consumption.
\end{abstract}

\section{Introduction}\label{sec:1}
For the last decade, high-speed railway (HSR) has become the future trend of railway transportation worldwide, and attracted a lot of attentions as a fast, convenient and green public transportation system.
With the continuous construction of HSR in recent years, the demand for mobile communication on high-speed trains is increasingly growing \cite{Ai-2014}.
More and more services related with the railway controlling information need to be transmitted between the train and the ground in order to guarantee the train moving safety.
Meanwhile, the passengers have an increasingly high demand on wireless Internet services when they are onboard.
To fulfill the high demand for wireless data transmission, the study on efficiency of HSR communications is critical.

There have been some recent works to improve the transmission performance in HSR communication systems.
From the network architecture perspective, a relay-assisted HSR network architecture has been proposed in \cite{Wang-2012} and \cite{Tian-2012}, which can provide better performance than direct transmission in case of large penetration loss.
To better utilize the network resources, \cite{Yan-2015} and \cite{Liang-2012} considered control/data signaling decoupled and cellular/infostation integrated HSR network architectures, respectively.
From the transmission technology perspective, the radio-over-fiber (RoF) technology for HSR communications was proposed in \cite{Lannoo-2007}, which can improve handover performance effectively.
Multi-input Multi-output technology (MIMO) was introduce into HSR scenarios in order to increase the network throughput \cite{Luo-2013}.
However, when considering the multiple services transmission between the train and the ground, more investigations on resource management are necessary to further improve the transmission performance.

In HSR communications, many types of services need to be transmitted between the train and the ground \cite{Pareit-2012}.
In particular, these HSR services are classified into four categories \cite{Brussel-2010}, i.e., pure passenger internet, passenger comfort services, security-related services and cost saving applications.
The effective transmission for these heterogenous services is a technical challenge.
First, the channel condition cannot remain at the same level due to the fast-varying distance  between the train and the ground, which causes that the power control along the time has a large influence on transmission performance.
Second, there exist heterogeneous quality-of-service (QoS) requirements in HSR communications, especially the end-to-end delay requirements since the security-related services should be delivered in time. The resource allocation plays a key role in enhancing the QoS performance by making full use of the limited resources.
Based on the above two aspects, the power control along the time and resource allocation among the delay-aware services in HSR wireless communications are still interesting and challenging problems.

To the best of our knowledge, resource allocation and power control in HSR communications are usually considered as separate problems.
In this paper, we jointly optimize them for delay-aware multi-service transmission in HSR communications.
Specifically, the main contributions are summarized as follows.
\begin{itemize}
  \item A stochastic optimization framework for multi-service transmission in HSR communication systems is developed, which focuses on dynamic resource management under the heterogeneous delay requirements and power constraints. The proposed framework is based on a cross-layer design to improve the efficiency of resource management, which involves the interactions between physical (PHY) layer and media access control (MAC) layer.
  \item Inspired by the stochastic network optimization approach, the intractable stochastic optimization problem is transformed into a tractable deterministic optimization problem. A static resource management algorithm is proposed to solve it with guaranteed global optimality, by using the golden section search method. Based on the static algorithm, we propose a dynamic resource management algorithm to solve the original stochastic optimization problem.
  \item The algorithm performance is evaluated by simulations under realistic assumptions for HSR communication systems. Simulation results show that compared with the traditional power control schemes, the proposed dynamic algorithm can effectively improve delay performance. In addition, we notice there exists a fundamental tradeoff between delay requirements and power consumption.
\end{itemize}

The remainder of the paper is structured as follows.
In Section \ref{sec:2}, we review the related works.
Section \ref{sec:3} describes the system model.
The problem formulation and transformation are provided in Section \ref{sec:4}.
We propose a dynamic resource management algorithm in Section \ref{sec:5}.
Numerical results and discussions are shown in Section \ref{sec:6}.
Finally, some conclusions are drawn in Section \ref{sec:7}.

Notations: In this paper, $\mathbb{E}[\cdot]$ denotes expectation. $\lfloor x \rfloor = \max \{n\in \mathbb{Z} | n \leq x\}$. $\lceil x \rceil = \min \{n\in \mathbb{Z} | n \geq x\}$. $\mathbb{R}$, $\mathbb{Z}$ and $\mathbb{N}$ denote the sets of real numbers, all integers and all positive integers, respectively.

\section{Related Work} \label{sec:2}

\subsection{Resource Allocation}
Resource allocation plays an important role in enhancing the data transmission efficiency and improving the QoS performance.
In the literature, the resource allocation problem in HSR communications has attracted great research interest.
For example, \cite{Zhao-2013} and \cite{Karimi-2012} investigated rate maximization resource allocation problem under the limited resource constraint.
Different energy-efficiency resource allocation methods were developed in \cite{Zhu-2012} and \cite{Zhao-2012} to minimize the total transmit power while satisfy QoS requirements.
However, a typical assumption in these works is the infinite backlog and the delay-insensitive services.
As a result, these works focus only on optimizing the PHY layer performance metrics such as sum throughput and total transmit power, and the resultant resource allocation schemes are adaptive to the channel condition only.

In practical HSR communications, it is important to focus on cross-layer optimization design, which considers random bursty arrivals and delay performance metrics in addition to the PHY layer performance metrics.
There is also plenty of literature on cross-layer resource optimization in HSR communications.
\cite{Liang-2012} and \cite{Chen-2014} investigated the resource allocation problem for delivering multiple on-demand services while considering their deadline constraints.
A cross-layer design approach was proposed in \cite{Zhu-2011} to improve video transmission quality by jointly optimizing application-layer parameters and handoff decisions.
In addition, \cite{Xu-2014-DADS} and \cite{Xu-2014-ICC} studied the downlink resource allocation problem with the delay constraint and packet delivery ratio requirement in relay-assisted HSR communications.
All the above works treat the resource allocation problem with the assumption of a constant transmit power.
When delivering multiple services between the ground and the train, the total resource allocated to the services is controlled by the instantaneous transmit power.
Thus, it is necessary to jointly consider the resource allocation among the services and the power control along the time.

\subsection{Power Control}
There are three unique features in HSR communications \cite{Lin-2012}, i.e., the deterministic moving direction, relatively steady moving speed and the accurate train location.
The data transmission rate is highly determined by the transmit power and the distance between the ground and the train, thus these features make it necessary and feasible to implement power control along the time.
Under the total power constraint, \cite{Dong-2014} presented four power allocation schemes to achieve different design objectives.
As an extension, \cite{Xu-2014-JWCN} investigated the utility-based resource allocation problem, which jointly considers the power allocation along the time and packet allocation among the services.
The delay-aware power allocation policy has been proposed in \cite{Zhang-2015} under the assumption of constant-rate data arrival.
Moreover, \cite{Thai-2013} and \cite{Ma-2012} studied the energy-efficient data transmission problem in HSR communications, with the purpose of minimizing the total transmit power.
However, these above works only take account of the time-varying channel state while do not consider the dynamic characteristics of the service or packet arrivals, which causes that the above power control schemes are not practical.

Dynamic power control is necessary to improve the performance of HSR communication systems, where the transmit power should be adaptive to the time-varying channel state and queue state.
The work \cite{Xu-2014-ACRA} investigated a joint admission control and resource allocation problem, which aims to maximize the system utility while stabilizing all transmission queues.
Different from \cite{Xu-2014-ACRA}, our work in this paper focuses on the delay-aware dynamic resource allocation and power control problem under power constraints.

\section{System Model}\label{sec:3}

As shown in Fig. \ref{SystemModel-DC}, we consider a HSR communication network consisting of a linear cellular network and a backbone network.
The linear cellular network deployed near the rail line can provide data transmission between the ground and the train.
In the backbone network, the distributed content servers (CSs) are deployed in order to offload data traffic \cite{Spagna-2013} and a central controller (CC) is responsible for resource management \cite{Liang-2012}.
The base stations (BSs) in the cellular network are connected to the CSs via wireline links, thus BSs and CSs can communicate with a negligible delay.
Considering the downlink transmission from the ground to the train, the data packets of the requested services are first delivered from the CSs to the vehicle station via the BSs, and then the vehicle station installed on the train transfers these data packets to the users on the train.
Since the communication between the ground and the vehicle station suffers from the fast-varying wireless channel and may become the bottleneck in this network architecture.
Therefore, this paper mainly considers the multi-service downlink transmission from the BSs to the vehicle station.

\begin{figure}[!htb]
    \centering
    \setlength{\abovecaptionskip}{0cm}
    \includegraphics[scale = 0.58]{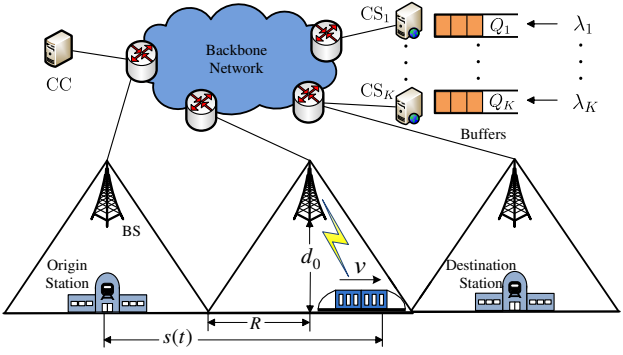}
    \caption{System model}\label{SystemModel-DC}
\end{figure}

\subsection{Deterministic Train Trajectory}
The deterministic train trajectory in HSR communication systems is a unique feature, which represents the train location at a specific time \cite{Liang-2012}.
The train trajectory information can be obtained accurately due to the following two reasons.
First, since the train moves on a predetermined rail line and the velocity is relatively steady.
Second, many train positioning techniques are applied into railway communications, such as the global positioning system (GPS) and digital track maps \cite{Pascoe-2009}.

In this paper, we develop a time-distance mapping function to describe the deterministic train trajectory.
As shown in Fig. \ref{SystemModel-DC}, we consider a train travels from the origin station to the destination station within the time duration $[0,T]$ at a constant speed $v$.
The whole time duration is partitioned into slots of equal duration $T_{s}$.
The distance which the train has traveled until slot $t$ is $s(t) = vt$ and the train location between two adjacent BSs is $s(t)~{\rm mod}~2R$, where $R$ is the cell radius. Define a time-distance mapping function $d(t): [0,T] \rightarrow [d_{0},d_{\rm max}]$, where $d_{\rm max} = \sqrt{R^{2} + d_{0}^{2}}$ and $d_{0}$ is the distance between each BS and the rail line. With the help of the train trajectory information, the distance $d(t)$ at slot $t$ can be obtained based on the geometry knowledge.

\subsection{PHY Layer Model}
The newly-built HSR routes are mainly composed of wide plain and viaduct, which yield a free space with few reflectors or scatterers.
Most of the time, only the direct signal path between BS and vehicle station is available, which was confirmed by engineering measurements \cite{Liu-2012, He-2011}.
Similar to \cite{Dong-2014, Zhang-2015, Zhang-2013}, we assume that the channel condition variation results only from the time-varying distance between BS and train.
Given the transmit power $P(t)$ and the distance $d(t)$,
the received signal-to-noise ratio (SNR) at slot $t$ can be denoted by
\begin{equation}\label{SNR}
  {\rm SNR}(t) = \frac{P(t)d^{-\alpha}(t)}{BN_{0}} = \frac{P(t)}{N(t)},
\end{equation}
where $B$ is the system bandwidth, $N_{0}$ is the noise power spectral density, $\alpha$ is the pathloss exponent, and
$N(t)$ is defined as $N(t) \triangleq BN_{0}d^{\alpha}(t)$ for brevity.

Based on the SNR expression \eqref{SNR}, the downlink transmission rate at slot $t$ is expressed by
\begin{equation}\label{transmission rate}
  R(t) = B\log_{2}\left(1+\frac{P(t)}{N(t)}\right)~\rm bits/s.
\end{equation}
Suppose that a packet is the transmission unit, hence the link capacity at slot $t$ can be denoted as the maximum number of packets, which can be expressed by
\begin{equation}\label{transmission rate 2}
  C(t) = \left\lfloor\frac{T_{s}R(t)}{L}\right\rfloor = \left\lfloor\frac{1}{\eta}\log_{2}\left(1+\frac{P(t)}{N(t)}\right)\right\rfloor,
\end{equation}
where $L$ is the packet size in bits, $\eta = \frac{L}{T_{s}B} > 0$, and $T_{s}$ is a slot duration.
Based on \eqref{transmission rate 2}, we can see that the link capacity $C(t)$ is determined by transmit power $P(t)$. Although the equal packet size and equal time duration at each slot are considered, the results presented herein can be extended to the unequal scenarios.

We consider the erasure coding based service transmission in PHY layer, which has been adopted in \cite{Liang-2012, Xu-2014-ICC} to simplify the protocol design for HSR communications.
The advantage is that no recovery scheme is required for the transmission error or loss of specific packets due to highly dynamic wireless channels.

\subsection{MAC Layer Model}
A set $\mathcal{K} \triangleq \{1,\ldots,K\}$ of delay-constrained services are supported over the trip.
We assume that ${\rm CS}_{k}$ is equipped with a buffer and can provide service $k$. Thus, we can see $K$ delay-constrained queues in MAC layer, as shown in Fig. \ref{SystemModel-DC}.
The maximum size of each buffer $Q_{{\rm max}}$ is assumed to be sufficiently large.
Let $\mathbf{Q}(t) = (Q_{1}(t),\ldots,Q_{K}(t))$ represent the current queue backlogs vector, where $Q_{k}(t)$ denotes the number of packets at the beginning of slot $t$ in the queue of ${\rm CS}_{k}$.
Let $\mathbf{A}(t) = (A_{1}(t),\ldots,A_{K}(t))$ represent the packet arrival vector, where $A_{k}(t)$ denotes the number of packets arriving into the buffer of ${\rm CS}_{k}$ at slot $t$.
Suppose in general, $A_{k}(t)$ follows Poisson distribution with average packet arrival rate $\lambda_{k}$ for service $k$.

The MAC layer is responsible for the resource allocation among the services.
Let $\boldsymbol{\mu}(t) = (\mu_{1}(t),\ldots,\mu_{K}(t))$ be the resource allocation action vector at slot $t$, where $\mu_{k}(t)$ is the number of packets allocated to service $k$.
Since the total number of allocated packets can not exceed the link capacity, the resource allocation at each slot $t$ must satisfy the constraint $0 \leq \sum_{k}\mu_{k}(t) \leq C(t)$.
In addition, the dynamics for all queues are given by
\begin{equation}\label{real queue dynamics}
  Q_{k}(t+1) = Q_{k}(t) - \mu_{k}(t) + A_{k}(t),~\forall k \in \mathcal{K}.
\end{equation}
Since the arrival packets at slot $t$ can only be transmitted after slot $t$, we have $0 \leq \mu_{k}(t) \leq Q_{k}(t),~ \forall k\in \mathcal{K}$.

\section{Problem Formulation and Transformation}\label{sec:4}
In this section, we first present a detailed description of the delay-aware multi-service transmission from the perspective of cross-layer design.
The average delay constraint and average power constraint are formulated in terms of long-term time average.
Then auxiliary variables are introduced to transform these long-term average constraints into the queue stability constraints.
Based on the Lyapunov drift theory, we formulate the delay-aware resource allocation and power control problem as a stochastic optimization problem.
Finally, the intractable stochastic optimization problem is transformed into a tractable deterministic optimization problem.

\subsection{Problem Statement}
This paper considers the delay-aware multi-service transmission in HSR communication systems,
with a focus on dynamic resource allocation and power control problem.
Based on the model in Section \ref{sec:3}, the problem can be stated as follows:
During a trip, considering system dynamic characteristics, i.e., the random packet arrivals and time-varying wireless channels, how to dynamically optimize resource allocation and power control to satisfy the heterogenous delay requirements of multiple services under power constraints along the time.

To enhance the efficiency of resource utilization and improve delay performance of service transmission, it is necessary to dynamically control resource in a cross-layer way.
Fig. \ref{CrossLayer} presents an illustration of cross-layer resource management, which involves the interactions between the PHY layer and the MAC layer.
At the PHY layer, the channel state information (CSI) allows an observation of good transmission opportunity.
At the MAC layer, the queue state information (QSI) provides the urgency of data packets.
The control actions, including power control action $P$ and resource allocation action vector $\boldsymbol{\mu}$, should be taken dynamically based on the PHY layer CSI and the MAC layer QSI.
Specifically, the power control action decides the link capacity, i.e., the total allocated packets. The resource allocation action decides how many packets are allocated for each service.

\begin{figure}[!h]
    \centering
    \includegraphics[scale = 0.53]{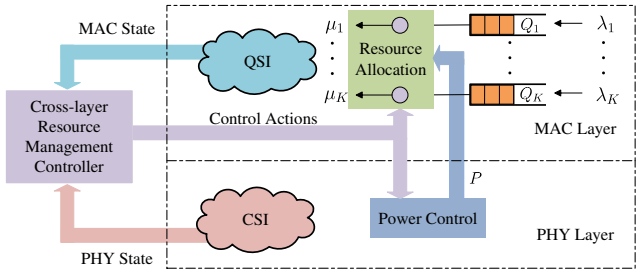}
    \caption{Cross-layer design for dynamic resource management }\label{CrossLayer}
\end{figure}

\subsection{Constraint Formulation}
We define the following notation for the long-term time average expectation of any quantity $x$,
\begin{equation}\label{long-term time average expectation}
  \overline{x} \triangleq \lim_{t\rightarrow\infty}\frac{1}{t}\sum_{\tau=0}^{t-1} \mathbb{E}[x(\tau)],
\end{equation}
Based on the definition \eqref{long-term time average expectation}, $\overline{Q}_{k}$ and $\overline{P}$ are denoted as average queue backlog for queue $k$ and average power consumption, respectively.

For the erasure coding based service delivery, the average delay constraint is considered since the decoding delay will be closely related to all encoded packets.
Mathematically, the average delay constraint for queue $k$ can be expressed by
\begin{equation}\label{average delay constraint}
  \overline{W}_{k} \leq W_{k}^{{\rm av}},
\end{equation}
where $\overline{W}_{k}$ and $W_{k}^{{\rm av}}$ represent the average delay and maximum average delay for queue $k$, respectively.
Based on Little's law, the average delay can be obtained by $\overline{W}_{k} = {\overline{Q}_{k}\over\lambda_{k}}$.
Thus, the constraint \eqref{average delay constraint} is equal to
\begin{equation}\label{average queue constraint}
  \overline{Q}_{k} \leq W_{k}^{{\rm av}} \lambda_{k}.
\end{equation}

The data transmission between the train and ground is subject to the transmit power constraints, including the maximum power constraint and average power constraint.
Mathematically, the maximum power constraint at any slot $t$ is given by $P(t) \leq P_{{\rm max}}$ and the average power constraint can be expressed by
\begin{equation}\label{average power constraint}
  \overline{P} \leq P_{{\rm av}},
\end{equation}
where $P_{{\rm av}}$ and $P_{{\rm max}}$ denote the maximum average power and the maximum instantaneous power, respectively.

The objective of this paper is to investigate the dynamic resource allocation and power control problem under the average delay constraints \eqref{average delay constraint}, the average power constraint \eqref{average power constraint} and the maximum power constraint.
In order to better characterize the considered problem, we consider the constraint transformation in the following subsection.

\subsection{Constraint Transformation}
To facilitate satisfaction of the constraint \eqref{average queue constraint}, we define a virtual queue $X_{k}(t)$ for each $k$ with the update equation
\begin{align}\label{virtual queue 2}
  X_{k}(t+1)  = \max[X_{k}(t) - W_{k}^{{\rm av}}\lambda_{k}, 0] + Q_{k}(t+1)
\end{align}
where $Q_{k}(t+1)$ is defined in (\ref{real queue dynamics}) and the initial condition is assumed $X_{k}(0) = 0$ for all $ k$.
Intuitively, $Q_{k}(t+1)$ and $W_{k}^{{\rm av}}\lambda_{k}$ can be viewed as the ``arrivals'' and the ``offered service'' of queue $X_{k}(t)$, respectively.

\begin{lemma}\label{L1}
If the virtual queue $X_{k}(t)$ is rate stable, i.e., satisfies $\lim\limits_{t \rightarrow \infty} \frac{X_{k}(t)}{t} = 0$, then $\overline{W}_{k} \leq W_{k}^{{\rm av}}$ holds and the queue $Q_{k}(t)$ is stable.
\end{lemma}

\begin{proof}
From \eqref{virtual queue 2}, we have
$X_{k}(\tau+1) \geq X_{k}(\tau) - W_{k}^{{\rm av}}\lambda_{k} + Q_{k}(\tau+1)$,
i.e., $X_{k}(\tau+1) - X_{k}(\tau) \geq Q_{k}(\tau+1) - W_{k}^{{\rm av}}\lambda_{k}$ for any $\tau$.
Summing the above over $\tau \in \{0,\ldots, t-1\}$ yields $X_{k}(t) - X_{k}(0) \geq \sum_{\tau=0}^{t-1} (Q_{k}(\tau+1) - W_{k}^{{\rm av}}\lambda_{k})$. Dividing by $t$ and taking limit as $t \rightarrow \infty$, we can get $\lim_{t \rightarrow \infty} X_{k}(t)/t \geq \overline{Q}_{k} - W_{k}^{{\rm av}}\lambda_{k}$. Thus, if $\lim_{t \rightarrow \infty} X_{k}(t)/t$ $ = 0$, then $\overline{Q}_{k} \leq W_{k}^{{\rm av}}\lambda_{k}$ holds, which implies the queue $Q_{k}(t)$ is stable.
In addition, based on Little's law, we have $\overline{Q}_{k} = \lambda_{k}\overline{W}_{k}$ and $\overline{W}_{k} \leq W_{k}^{{\rm av}}$ for $k \in \mathcal{K}$.   
\end{proof}

The intuition behind Lemma \ref{L1} is that if the excess backlog in the virtual queue is stabilized, it must be the case that the time average arrival rate $\overline{Q}_{k}$ is not larger than the service rate $W_{k}^{\rm av}\lambda_{k}$.
Based on the Lemma \ref{L1}, the constraint (\ref{average delay constraint}) can be transformed into a single queue stability problem.

Similarly, for the constraint (\ref{average power constraint}), we define the virtual queue $Y_{k}(t)$ for each $k$ with the update equation
\begin{equation}\label{virtual queue 3}
  Y_{k}(t+1) = \max[Y_{k}(t) - P_{{\rm av}}, 0] + P(t).
\end{equation}
Thus, stabilizing $Y_{k}(t)$ ensures $\overline{P} \leq P_{{\rm av}}$.

\subsection{Problem Formulation}
Define $\mathbf{X}(t)$ and $\mathbf{Y}(t)$ as a vector of all virtual queues $X_{k}(t)$ and $Y_{k}(t)$, respectively.
We denote $\boldsymbol{\Theta}(t)$ as the combined vector of all virtual queues,
$\boldsymbol{\Theta}(t) \triangleq [\mathbf{X}(t),\mathbf{Y}(t)]$.
Define the quadratic Lyapunov function \cite{Neely-2010}
\begin{equation}\label{Lyapunov function}
  L(\boldsymbol{\Theta}(t)) \triangleq \frac{1}{2} \left( \sum_{k \in \mathcal{K}} X_{k}(t)^{2} + \omega\sum_{k \in \mathcal{K}} Y_{k}(t)^{2} \right),
\end{equation}
where $\omega \geq 0$ represents the weight on how much we emphasize the average power constraint.

Next, $\Delta(\boldsymbol{\Theta}(t))$ is defined as the one-slot conditional Lyapunov drift at slot $t$
\begin{equation}\label{Lyapunov drift}
  \Delta(\boldsymbol{\Theta}(t)) = \mathbb{E}[L(\boldsymbol{\Theta}(t+1)) - L(\boldsymbol{\Theta}(t))|\boldsymbol{\Theta}(t)],
\end{equation}
which can help to ensure that the virtual queues are stable and the desired constraints are met.
At each slot $t$, observing the virtual queue vector $\boldsymbol{\Theta}(t)$ and real queue vector $\mathbf{Q}(t)$, the resource allocation action vector $\boldsymbol{\mu}(t)$ and power control action $P(t)$ should be jointly decided to minimize the drift \eqref{Lyapunov drift}.
Thus, the resource management problem at slot $t$ is formulated as
\begin{subequations}\label{subeq:1}
\begin{align}
               \min\limits_{P(t), \boldsymbol{\mu}(t)}~~&~~\Delta(\boldsymbol{\Theta}(t))\label{maximize}\\
               \mbox{s.t.}~~&~~0 \leq P(t) \leq P_{{\rm max}}\label{max power constraint}\\
               ~~&~~0 \leq \mu_{k}(t) \leq Q_{k}(t), ~\mu_{k}(t)\in \mathbb{N}, ~\forall k \label{RA constraint 1}\\
               ~~&~~\sum_{k \in \mathcal{K}}\mu_{k}(t) \leq C(t) \label{RA constraint 2}
\end{align}
\end{subequations}

The problem \eqref{subeq:1} is a stochastic optimization problem \cite{Neely-2010}, but it cannot be solved efficiently since the difficulty from the form of the objective function (\ref{maximize}). In order to better characterize the
problem \eqref{subeq:1} and develop an efficient algorithm to solve it, we consider the problem transformation in the following subsection.

\subsection{Problem Transformation}
To make the objective function \eqref{maximize} easily handled, we have the following lemma.
\begin{lemma}\label{L2}
Under any $\boldsymbol{\mu}(t)$, $P(t)$ and $\boldsymbol{\Theta}(t)$, we have
\begin{equation}\label{L2-1}
  \Delta(\boldsymbol{\Theta}(t)) \leq {1 \over 2}D + \mathbb{E}[G(t)|\boldsymbol{\Theta}(t)],
\end{equation}
where $D$ is a finite constant defined as
\begin{equation}\label{L2-2}
  D = \sum_{k \in \mathcal{K}}\Big[Q_{{\rm max}}^{2} + (\lambda_{k}W_{k}^{{\rm av}})^{2} + \omega (P_{{\rm max}}^{2} + P_{{\rm av}}^{2}) \Big],
\end{equation}
and $G(t)$ is defined as
\begin{align}\label{L2-3}
G(t)\triangleq & \sum_{k \in \mathcal{K}} \Big[X_{k}(t)\big(Q_{k}(t) - \mu_{k}(t) + A_{k}(t) - \lambda_{k}W_{k}^{{\rm av}}\big) \nonumber \\
& + \omega Y_{k}(t)(P(t)-P_{{\rm av}}) \Big].
\end{align}
\end{lemma}

\begin{proof}
By squaring the equation \eqref{virtual queue 2}, we have
\begin{eqnarray}\label{PL2-1}
   &&\!\!\!\! X_{k}(t+1)^{2} - X_{k}(t)^{2}  \nonumber\\
   &&\!\!\!\! = (\max[X_{k}(t) - \lambda_{k}W_{k}^{{\rm av}}, 0] + Q_{k}(t+1))^{2} - X_{k}(t)^{2} \nonumber\\
   &&\!\!\!\! \leq Q_{k}(t+1)^{2}+(\lambda_{k}W_{k}^{{\rm av}})^{2} + 2X_{k}(t)(Q_{k}(t+1)-\lambda_{k}W_{k}^{{\rm av}}) \nonumber\\
\end{eqnarray}
where we use the fact that for any $x,y\geq 0$, $(\max[x,0])^{2}$ $\leq x^{2}$ and $\max[x-y,0] \leq x$ in the inequality.

Similarly, it can be shown for any $k \in \mathcal{K}$
\begin{equation}\label{PL2-2}
  Y_{k}(t+1)^{2} - Y_{k}(t)^{2} \leq P(t)^{2} + P_{{\rm av}}^{2} + 2Y_{k}(t)(P(t)-P_{{\rm av}})
\end{equation}
Based on \eqref{Lyapunov drift}, \eqref{PL2-1} and \eqref{PL2-2}, we have
\begin{align}\label{PL2-4}
  &\Delta(\boldsymbol{\Theta}(t)\nonumber\\
  &= \mathbb{E}\Bigg[\frac{1}{2} \sum_{k\in \mathcal{K}} \big[X_{k}(t+1)^{2} - X_{k}(t)^{2}+ \omega Y_{k}(t+1)^{2} - \omega Y_{k}(t)^{2}\big]|\boldsymbol{\Theta}(t)\Bigg]\nonumber\\
  &\leq \mathbb{E}\Bigg[\frac{1}{2} \sum_{k\in \mathcal{K}} \Big[ Q_{k}(t+1)^{2}+(\lambda_{k}W_{k}^{{\rm av}})^{2} + \omega (P(t)^{2} + P_{{\rm av}}^{2}) \nonumber\\
  &+ 2X_{k}(t)(Q_{k}(t) - \mu_{k}(t) + A_{k}(t) -\lambda_{k}W_{k}^{{\rm av}})\nonumber\\
  &+ 2\omega Y_{k}(t)(P(t)-P_{{\rm av}}) \Big] |\boldsymbol{\Theta}(t) \Bigg]\nonumber\\
  &\leq D + \mathbb{E}\left[ G(t) |\boldsymbol{\Theta}(t) \right]
\end{align}
where $G(t)$ is defined by (\ref{L2-3}) and the last inequality can be obtained by
\begin{eqnarray}\label{PL2-5}
&&\mathbb{E}\Bigg[\sum_{k \in \mathcal{K}}\Big[Q_{k}(t+1)^{2} + (\lambda_{k}W_{k}^{{\rm av}})^{2} +\omega (P(t)^{2} + P_{{\rm av}}^{2})\Big] |\boldsymbol{\Theta}(t) \Bigg] \nonumber\\
&& \leq \mathbb{E}\Bigg[\sum_{k \in \mathcal{K}}\Big[Q_{{\rm max}}^{2} + (\lambda_{k}W_{k}^{{\rm av}})^{2} +\omega (P_{{\rm max}}^{2} + P_{{\rm av}}^{2}) \Big] |\boldsymbol{\Theta}(t) \Bigg]\nonumber\\
&& = \sum_{k \in \mathcal{K}}\Big[Q_{{\rm max}}^{2} + (\lambda_{k}W_{k}^{{\rm av}})^{2} + \omega (P_{{\rm max}}^{2} + P_{{\rm av}}^{2}) \Big] = D
\end{eqnarray}
where the inequality holds based on $Q_{k}(t+1) \leq Q_{{\rm max}}$ and $P(t) \leq P_{{\rm max}}$, and the equality holds since the constant in the square bracket is independent of $\boldsymbol{\Theta}(t)$.
\end{proof}

Based on Lemma \ref{L2}, the problem \eqref{subeq:1} can be simplified to minimize the drift upper bound, i.e., the right-hand-side of inequality \eqref{L2-1}.
We notice that the control actions are independent of the first term and only affect the second term on the right-hand-side of the inequality \eqref{L2-1}.
Thus, the objective turns to the minimization of the expression $\mathbb{E}[G(t)|\boldsymbol{\Theta}(t)]$.
This conditional expectation is with respect to the virtual queue vector $\boldsymbol{\Theta}(t)$ and the possible control actions.
Then, using the concept of opportunistically minimizing an expectation \cite{Neely-2010}, the control actions are chosen to minimize $G(t)$ by observing $\boldsymbol{\Theta}(t)$ and $\mathbf{Q}(t)$ at each slot $t$.
Next, isolating $\boldsymbol{\mu}(t)$ and $P(t)$ in \eqref{L2-3} leads to the following expression
\begin{align}\label{drift-plus-penalty-1}
      \sum_{k \in \mathcal{K}} \Big[\omega Y_{k}(t)P(t) - X_{k}(t)\mu_{k}(t)\Big].
\end{align}
Therefore, the intractable stochastic optimization problem \eqref{subeq:1} can be transformed into a deterministic optimization problem at each slot, which is expressed by
\begin{subequations}\label{subeq:2}
      \begin{align}
      \max\limits_{P, \boldsymbol{\mu}}~~&\sum_{k \in \mathcal{K}} \Big[X_{k}\mu_{k} - \omega Y_{k}P \Big] \label{subeq:2-1}\\
      \mbox{s.t.}~~&~~0 \leq P \leq P_{{\rm max}}\label{subeq:2-1}\\
       ~~&~~0 \leq \mu_{k} \leq Q_{k}, ~\mu_{k}\in \mathbb{N}, ~\forall k\label{subeq:2-2}\\
               ~~&~~\sum_{k \in \mathcal{K}}\mu_{k} \leq C= \left\lfloor\frac{1}{\eta}\log_{2}\Big(1+\frac{P}{N}\Big)\right\rfloor\label{subeq:2-3}
      \end{align}
\end{subequations}
Note that the time index is omitted in problem \eqref{subeq:2} for brevity.
Let $\boldsymbol{\mu}^{\ast}=(\mu_{1}^{\ast},\ldots,\mu_{K}^{\ast})$ and $P^{\ast}$ denote the optimal resource allocation action vector and the optimal power control action for problem \eqref{subeq:2}, respectively.

\section{Dynamic Resource Management Algorithm}\label{sec:5}
The problem \eqref{subeq:2} is a mixed integer programming (MIP) problem. A common way of solving it is to relax the integer constraints and then the problem becomes a convex optimization problem, which can be solved by CVX \cite{Grant-2014}. In addition, the optimization solvers, such as CPLEX and LINDO, have been successfully applied to MIP problems.
However, the above methods often have a high computational complexity.
Thus, to overcome this disadvantage, we consider the problem transformation and then propose a static resource management algorithm to effectively solve problem \eqref{subeq:2}. Finally, a dynamic resource management algorithm is proposed to solve the original problem \eqref{subeq:1}.

\subsection{Static Resource Management Algorithm}
In this subsection, the problem \eqref{subeq:2} is equivalently transformed into a single variable problem, which will be discussed below.
First, we focus on analyzing the constraints in problem \eqref{subeq:2}.
Specifically, we notice that the optimal solution will always achieve the equality in constraint \eqref{subeq:2-3}, which can be given by
\begin{equation}\label{equality}
  \sum_{k \in \mathcal{K}}\mu_{k} = C = \frac{1}{\eta}\log_{2}\Big(1+\frac{P}{N}\Big).
\end{equation}
Otherwise we can reduce the value of $P$ so as to increase the objective value without any violation of the constraints (\ref{subeq:2-1}) and (\ref{subeq:2-2}).
Based on the first equality in \eqref{equality} and the constraint \eqref{subeq:2-2}, the link capacity $C$ should
satisfy $0 \leq C \leq \sum_{k}Q_k$.
In addition, since the link capacity $C$ is the sum of integers, it is also an integer, i.e., $C \in \mathbb{N}$.
From the second equality in \eqref{equality}, there exists a one-to-one relationship between $P$ and $C$.
Thus, the power consumption $P$ can be expressed by
\begin{equation}\label{one-to-one relationship}
  P = N(2^{\eta C}-1).
\end{equation}
Based on \eqref{one-to-one relationship}, the constraint (\ref{subeq:2-1}) is equivalent to $0 \leq C \leq C_{\text{max}}$,
where $C_{\text{max}} \triangleq \frac{1}{\eta} \log_{2}\left(1+\frac{P_{\text{max}}}{N}\right)$.
From the above analysis, the link capacity $C$ should satisfy
\begin{equation}\label{sub3-3}
  0 \leq C \leq \min\left(\sum_{k\in\mathcal{K}}Q_k, C_{\text{max}}\right).
\end{equation}

Then, the problem \eqref{subeq:2} can be transformed into a single variable problem as shown below
\begin{subequations}\label{sub3-4}
\begin{align}
      \max\limits_{C \in \mathbb{N}}~& M(C) \triangleq M_{1}(C) - M_{2}(C)\label{sub3-4-1}\\
      \mbox{s.t.}~&\eqref{sub3-3}, \label{sub3-4-2}
\end{align}
\end{subequations}
where $M_{1}(C)$ is given by
\begin{subequations}\label{sub3-5}
\begin{align}
      M_{1}(C) \triangleq \max\limits_{\boldsymbol{\mu}}~& \sum_{k \in \mathcal{K}} X_{k}\mu_{k}\label{sub3-5-1}\\
      \mbox{s.t.}~&~ 0\leq \mu_{k} \leq Q_{k},~\mu_{k}\in \mathbb{N},~\forall k \in \mathcal{K},\label{sub3-5-2}\\
      ~&\sum_{k \in \mathcal{K}}\mu_{k} = C, \label{sub3-5-3}
\end{align}
\end{subequations}
and $M_{2}(C)$ is given by
\begin{equation}\label{sub3-6}
  M_{2}(C) \triangleq \sum_{k \in \mathcal{K}} \omega Y_{k}P = \beta \left(2^{\eta C}-1\right),
\end{equation}
with $\beta \triangleq \omega N\sum_{k \in \mathcal{K}} Y_k$.

\begin{lemma}\label{L3}
Problems \eqref{subeq:2} and \eqref{sub3-4} are equivalent.
\end{lemma}

\begin{proof}
We prove the equivalence from both the objective function and the constraints.
On one hand, we can observe that the objective functions of problem \eqref{subeq:2} and problem \eqref{sub3-4} are same, although the objective function of problem \eqref{sub3-4} is divided into two parts, i.e., \eqref{sub3-5-1} and \eqref{sub3-6}.
On the other hand, as for the constraints on the variable $\mu_{k}$, the constraint \eqref{subeq:2-2} is the same as the constraint \eqref{sub3-5-2} and the constraint \eqref{subeq:2-3} is equivalent to \eqref{sub3-5-3} based on the necessary condition of optimality \eqref{equality}.
As for the constraints on the variable $P$, based on the expression \eqref{one-to-one relationship}, the constraint $0 \leq P \leq P_{{\rm max}}$ in \eqref{subeq:2-1} is equivalent to $0 \leq C \leq C_{\text{max}}$ in \eqref{sub3-3}.
From the above analysis, we can conclude that problem \eqref{subeq:2} is equivalent to problem \eqref{sub3-4}.
\end{proof}

Thus, we can solve the problem \eqref{sub3-4} instead of \eqref{subeq:2}.
Let $C^{\ast}$ denote the optimal solution of the problem \eqref{sub3-4}.
Before solving the problem \eqref{sub3-4}, we first focus on the subproblem \eqref{sub3-5} with any given $C$.
It is worth noting that the maximum value of $M_{1}(C)$ can always be achieved by allocating link capacity $C$ to the services in the descending order of $X_{k}$.
For convenience, all the services are sorted in descending order of $X_{k}$ with the set $\{k_{1},k_{2},\ldots,k_{K}\}$.
Mathematically, the optimal solution to the subproblem \eqref{sub3-5} is given by
\begin{equation}\label{sub2-3}
  \mu_{k_{n}} = \min\left\{\max\left\{C- \sum_{m=0}^{n-1} Q_{k_{m}}, 0\right\}, Q_{k_{n}}\right\},~\forall n,
\end{equation}
where $Q_{k_{0}} =0$.

After solving the subproblem \eqref{sub3-5}, we focus on how to solve the problem \eqref{sub3-4}.
We relax $C \in \mathbb{N}$ to $C \in \mathbb{R}$ in problem \eqref{sub3-4}, and then the property of the objective function $M(C)$ will be exploited in the following lemma.

\begin{lemma}\label{L4}
$M(C)$ is concave over $[0, \sum_{m=0}^{K} Q_{k_{m}}]$.
\end{lemma}

\begin{proof}
On one hand, for a sufficiently small $\delta>0$, since $\Delta M_{1}(C) = M_{1}(C+\delta) -M_{1}(C) = \delta X_{k_{n}}$ for $\sum_{m=0}^{n-1} Q_{k_{m}} \leq C < \sum_{m=0}^{n} Q_{k_{m}}, \forall n \in [1, K]$, $\Delta M_{1}(C)$ is a non-increasing function of $C$.
On the other hand, since $\Delta M_{2}(C) = M_{2}(C+\delta) -M_{2}(C) = \beta 2^{\eta C}(2^{\eta\delta}-1)$, $\Delta M_{2}(C)$ is a monotonically increasing function of $C$.
Therefore, $\Delta M(C) = \Delta M_{1}(C) - \Delta M_{2}(C)$ is a monotonically decreasing function of $C$, which implies that $M(C)$ is concave over $[0, \sum_{m=0}^{K} Q_{k_{m}}]$.
\end{proof}

Based on \cite{Bertin-1984}, if $M(C)$ is concave, then $M(C)$ is unimodal.
Since $M(C)$ is an unimodal function of $C$ over $[0, \sum_{m=0}^{K} Q_{k_{m}}]$, the golden section search method \cite{Jacoby-1972} is very suitable for searching without derivative for the maximum of objective function $M(C)$ with unimodal.
Then the static resource management algorithm is proposed based on the golden section search method, as described in Algorithm \ref{a-1}.

\begin{algorithm}[!thb]
\caption{Static Resource Management Algorithm}\label{a-1}
\begin{algorithmic}[1]
\REQUIRE $X_{k}$, $Q_{k}$, $\beta$, the golden ratio $\varphi = \frac{\sqrt{5}-1}{2}$;
\STATE Initialize two endpoints, i.e., $\hat{C} = 0$ and $\check{C} = \min\{C_{\text{max}}, \sum_{m=0}^{K} Q_{k_{m}}\}$;
\STATE Determine two intermediate points $C_{1}$ and $C_{2}$ such that $C_{1} = \hat{C} + \varphi(\check{C} - \hat{C})$ and $C_{2} = \check{C} - \varphi(\check{C} - \hat{C})$;
\WHILE{$\check{C} - \hat{C} > \varepsilon$}
\STATE Obtain $M_{1}(C_{1})$ and $M_{1}(C_{2})$ by solving \eqref{sub3-5}, respectively;
\STATE Compute $M_{2}(C_{1})$ and $M_{2}(C_{2})$ based on \eqref{sub3-6};
\STATE $M(C_{1}) = M_{1}(C_{1}) - M_{2}(C_{1})$, $M(C_{2}) = M_{1}(C_{2}) - M_{2}(C_{2})$;
\IF{$M(C_{1}) \geq M(C_{2})$}
\STATE{$\hat{C} := C_{2}$, $C_{2} := C_{1}$, $C_{1} := \hat{C} + \varphi(\check{C} - \hat{C})$;}
\ELSE
\STATE{$\check{C}:= C_{1}$, $C_{1} := C_{2}$, $C_{2} := \check{C} + \varphi(\check{C} - \hat{C})$};
\ENDIF
\ENDWHILE
\STATE{$\tilde{C} := \frac{1}{2}(\check{C} + \hat{C})$;}
\STATE{Obtain $C^{\ast}$ by solving \eqref{optimal solution C};}
\STATE{Calculate $P^{\ast}$ by \eqref{one-to-one relationship} when $C=C^{\ast}$;}
\STATE{Obtain $\mu_{k}^{\ast}$ by \eqref{sub2-3} when $C=C^{\ast}$;}
\ENSURE $P^{\ast}$, $\boldsymbol{\mu}^{\ast}$
\end{algorithmic}
\end{algorithm}

From step 1 to step 13, the golden section search method is used to get the optimal solution while relaxing $C$ as a positive real number in problem \eqref{sub3-4}.
Specifically, two endpoints and two intermediate points in the search region are determined in step 1 and step 2, respectively.
The iterative calculation of the golden section search is implemented from step 3 to step 12 until $\check{C} - \hat{C} \leq \varepsilon$, where $\varepsilon$ is the iterative accuracy.
In each iteration, $M(C_{1})$ and $M(C_{2})$ are calculated from step 4 to step 6 and then are evaluated from step 7 to step 11.
The endpoints and intermediate points are updated during the evaluation.
When the iteration converges, the optimal solution $\tilde{C}$ after relaxation is obtained in step 13.
In step 14, the optimal solution of problem \eqref{sub3-4} is obtained by considering the integer nature of optimal variable $C$.
Due to the concavity of $M(C)$ based on Lemma \ref{L4}, $M(C)$ has a non-negative slope at $C = \lfloor\tilde{C}\rfloor$. Thus, $M(\lfloor\tilde{C}\rfloor) \geq M(\lfloor\tilde{C}\rfloor - \epsilon)$ for $0 \leq \epsilon \leq \tilde{C}$.
Similarly, we have $M(\lceil\tilde{C}\rceil) \geq M(\lceil\tilde{C}\rceil + \epsilon)$ for $0 \leq \epsilon \leq \tilde{C}$ due to the non-positive slope of $M(C)$ at $C = \lceil\tilde{C}\rceil$.
Thus, the optimal integer value of $C$ is either $\lfloor\tilde{C}\rfloor$ or $\lceil\tilde{C}\rceil$. Mathematically, it can be obtained by
\begin{equation}\label{optimal solution C}
  C^{\ast} = \argmax \limits_{C \in \{\lfloor\tilde{C}\rfloor,\lceil\tilde{C}\rceil\}} M(C).
\end{equation}
Finally, the optimal solutions $P^{\ast}$ and $\boldsymbol{\mu}^{\ast}$ can be obtained in step 15 and step 16, respectively.

\begin{figure*}[!htb]
  \includegraphics[scale=1.08]{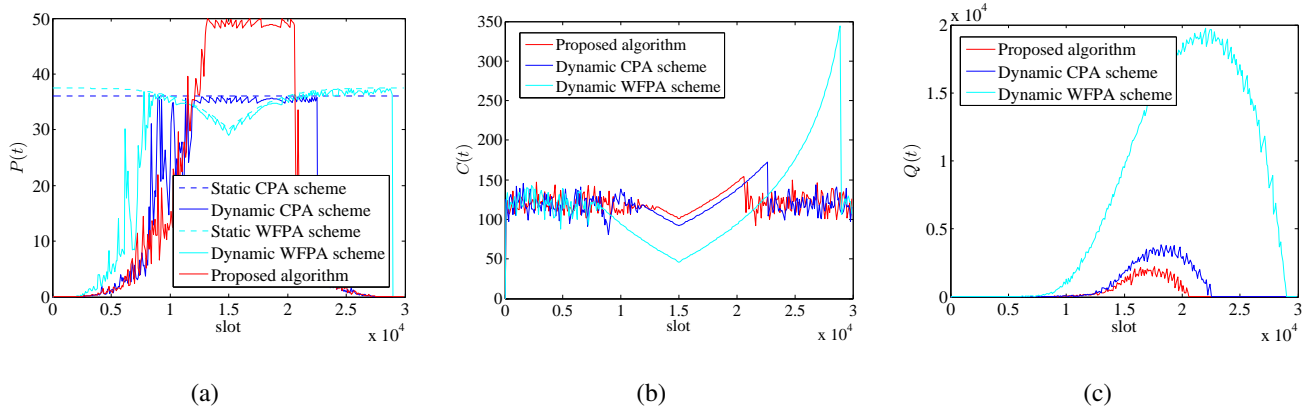}
\caption{The dynamic performance along the time, where $\lambda_{k}=20$ packets/slot, $W_{k}^{\rm av} =15$ slots, $P_{{\rm max}} = 50~W$ and $\omega = 0.8$.}
\label{CPQ}
\end{figure*}

\subsection{Dynamic Resource Management Algorithm}
In this subsection, we propose a dynamic resource management algorithm to solve the original problem \eqref{subeq:1} based on the static resource management algorithm.
Specifically, by observing the queue states at each slot, the dynamic algorithm is designed to choose control actions via solving problem \eqref{sub3-4}.
The detailed steps are described in Algorithm \ref{a-3}.
All system parameters should be initialized before the transmission process begins.
At the beginning of each slot, the problem \eqref{sub3-4} is solved by calling Algorithm \ref{a-1}.
At the end of each slot, the queues $Q_{k}(t+1)$, $X_{k}(t+1)$, and $Y_{k}(t+1)$ are updated according to \eqref{real queue dynamics}, \eqref{virtual queue 2} and \eqref{virtual queue 3}, respectively.
The algorithm will be repeated until all service transmissions are finished.

\begin{algorithm}[!thb]
\caption{Dynamic Resource Management Algorithm}\label{a-3}
\begin{algorithmic}[1]
\STATE Initialize $T_{s}$, $B$, $N_{0}$, $\omega$, $\eta$, $Q_{k}(0) = X_{k}(0) = Y_{k}(0) = 0$ for all $k$;
\STATE Obtain the trajectory information $d(t)$;
\FOR{$t = 0$ to $T$}
\STATE Calculate $\beta(t)$, $N(t)$, and $C_{\text{max}}(t)$;
\STATE Obtain $P(t)$ and $\boldsymbol{\mu}(t)$ by calling Algorithm \ref{a-1};
\STATE Update $Q_{k}(t+1)$, $X_{k}(t+1)$, and $Y_{k}(t+1)$ according to \eqref{real queue dynamics}, \eqref{virtual queue 2}, and \eqref{virtual queue 3}, respectively;
\ENDFOR
\end{algorithmic}
\end{algorithm}

\section{Simulation Results and Discussions}\label{sec:6}

\subsection{Simulation Setup}
We consider a real train schedule based on the Huhang high-speed railway \cite{Liang-2012, Chen-2014}.
The simulations in this paper are built on the train G7302 and the train trajectory is generated according to the mobility model proposed in \cite{Ahmad-2010}.
The other parameters are summarized in Table \ref{Table1}.

\begin{table}[!htb]
\renewcommand{\arraystretch}{1.0}
\caption{Simulation parameters}\label{Table1}
\centering
\begin{tabular}{lll}
\hline\noalign{\smallskip}
Parameter & Description & Value  \\
\noalign{\smallskip}\hline\noalign{\smallskip}
$P_{\text{av}}$& maximum average power & 36 W \\
$B$& system bandwidth &$5$ MHz \\
$L$& packet size &$240$ bits \\
$T_{s}$&slot duration&$1$ ms \\
$\alpha$& pathloss exponent &$4$ \\
$N_{0}$& noise power spectral density & -174 dBm/Hz \\
$v$& constant moving speed &$360$ km/h \\
$R$& cell radius & 1.5 km \\
$d_{0}$& distance between BS and rail &50 m \\
$K$& number of services&$6$\\
\noalign{\smallskip}\hline
\end{tabular}
\end{table}

For the purpose of comparison, we evaluate two related static power allocation schemes as reference benchmarks,
i.e., constant power allocation (CPA) scheme and water filling power allocation (WFPA) scheme \cite{Dong-2014}.
In the static CPA scheme, BS maintains a constant transmit power at all times, i.e., $P(t) = P_{\rm av}$.
In the static WFPA scheme, the water filling method is used to maximize the total throughput along the time.
Since the power allocation in these two static schemes is determined in advance, we modify them to the corresponding dynamic schemes in order to adapt to the variations of data traffic and channel state.
Specifically, the dynamic schemes can be obtained by replacing the maximum power $P_{\text{max}}$ in proposed dynamic algorithm with the static power allocation results.
The resultant dynamic schemes are denoted as dynamic CPA scheme and dynamic WFPA scheme, respectively.

\subsection{Performance Comparison}
Fig. \ref{CPQ} shows the dynamic performance along the time for different power allocation schemes.
For the sake of performance comparison, we only plot the simulation results during a time period when the train moves from the center of one cell to that of adjacent cell.

Fig. \ref{CPQ}(a) shows the power allocation along the travel time for different schemes.
We can see that the predetermined power allocation results in the static CPA scheme and static WFPA scheme are independent of random packet arrivals. Considering the packet arrival process, the transmit power changes dynamically in the three dynamic schemes.
Specifically, when the train moves towards the cell edge, the power consumption increases in all dynamic schemes since the wireless link quality degrades.
When the train moves at the cell edge (from $1 \times 10^{4}$th slot to $2 \times 10^{4}$th slot), nearly maximum transmit power is consumed in the proposed algorithm, while the power consumptions in the dynamic CPA scheme and dynamic WFPA scheme are limited by the predetermined power allocation.

Fig. \ref{CPQ}(b) and \ref{CPQ}(c) show the instantaneous link capacity and average queue backlog of services for the three dynamic schemes, respectively.
It can be observed in Fig. \ref{CPQ}(b) that when the train moves at the cell center (from $0$ to $10^{4}$th slot), the link capacity is just around the total packet arrivals 120 packets/slot for all schemes, which results in a small queue backlog shown in Fig. \ref{CPQ}(c).
This result can be explained as follows: Since the channel condition is good at the cell center, little power will be consumed shown in Fig. \ref{CPQ}(a). Thus, the packets can be transmitted immediately once they arrive at the buffer, which implies that the queue backlog is small.
As the train moves far from the cell center, the channel condition turns bad and much power will be consumed for transmitting one packet. The link capacity decreases for all schemes due to the power constraint, which causes the increasing queue backlog.
As shown in Fig. \ref{CPQ}(b), when the train locates near the cell edge, the link capacity in the proposed algorithm is more than that in the other two dynamic schemes due to the different power allocations in Fig. \ref{CPQ}(a).
The difference in the link capacity results in different queue backlogs among these dynamic schemes.
Furthermore, we can see from Fig. \ref{CPQ}(b) that the three curves suddenly drop after steady increasing, which implies the queue backlogs have been emptied. Compared with the other two schemes, less time will be spent on emptying queue backlogs and less buffer size is needed in the proposed algorithm.

Fig. \ref{PDA} shows the average power consumption and average delay performance with different packet arrival rates, respectively.
As expected, the average power consumption in all the schemes increases with the average packet arrival rates. Transmitting more packets with the same delay requirement will lead to more power consumption. However, the increment in the proposed algorithm is large while that in the other two schemes is small.
This is because that the power consumptions in the dynamic CPA scheme and dynamic WFPA scheme are limited by the predetermined power allocation.
From Fig. \ref{PDA}(b), we can see that the average delay in all the schemes also increases with the average packet arrival rates. When the packet arrival rate increases, the queue backlog gets larger, which further results in longer queue delay.
Furthermore, we observe that as for the same packet arrival rate, the average delay in the proposed algorithm is much lower than the other two schemes.
Specifically, when the packet arrival rate is 25 packets/slot, the average delay in the proposed algorithm can respectively be 6.3$\%$ and 22.2$\%$ of that in the other two schemes, which demonstrates that the proposed algorithm outperforms the other two schemes in term of delay performance.


\begin{figure}[!htb]
  \centering
  \subfigure[]{
    \label{P-Arrival} 
    \includegraphics[width=0.95\linewidth]{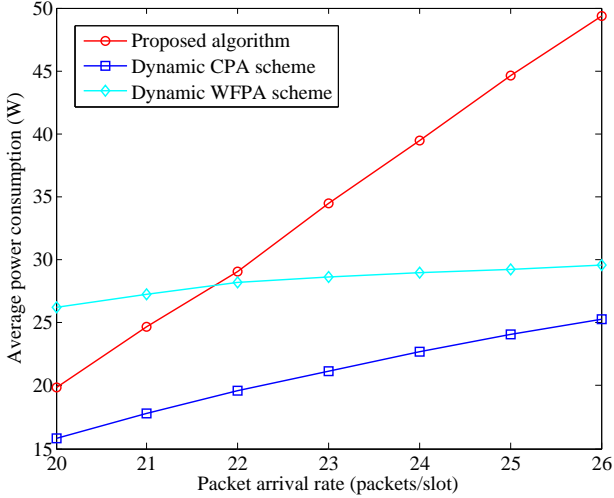}}
  \subfigure[]{
    \label{D-Arrival} 
    \includegraphics[width=0.95\linewidth]{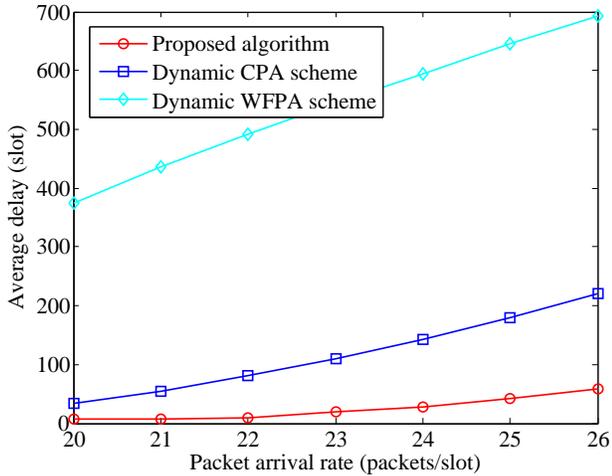}}
  \caption{Average power consumption and average delay under different packet arrival rates, where $\lambda_{k}=20$ packets/slot, $W_{k}^{\rm av} =15$ slots, $P_{{\rm max}} = 100~W$ and $\omega = 0.8$.}
  \label{PDA}
\end{figure}

\begin{figure}[!tbh]
    \centering
    \setlength{\abovecaptionskip}{-0.06cm}
    \includegraphics[scale = 0.48]{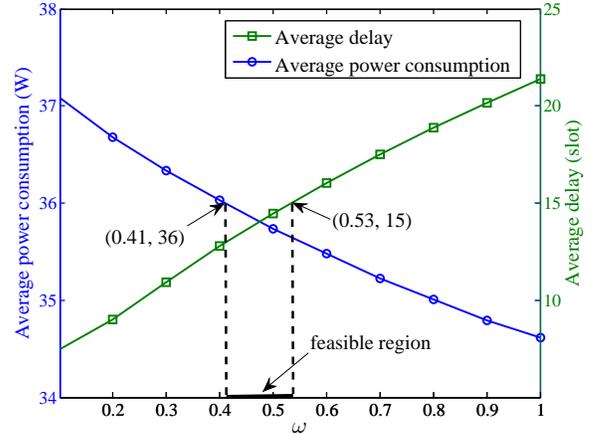}
    \caption{Average power consumption and average delay under different weights, where $\lambda_{k} = 23$ packets/slot, $W_{k}^{{\rm av}} = 15$ slots, and $P_{\rm max} = 100~W$.}\label{P-D}
\end{figure}

\subsection{Effects of the Parameter $\omega$}
We show the effects of the weight $\omega$ on the power consumption and delay performance in the proposed algorithm.
Fig. \ref{P-D} plots the average power consumption and average delay versus the weight $\omega$.
It can be seen that the average power consumption decreases with increasing $\omega$ while the average delay increases with increasing $\omega$. The reason is that increasing $\omega$ leads to more weight putting on average power consumption constraint and hence less power is utilized for transmitting buffered packets, which results in longer queue delay.
Thus, we can see that $\omega$ plays a key role in balancing the average delay and average power consumption.
Furthermore, to satisfy the average delay constraint and average power constraint simultaneously, it is necessary to find the reasonable range of $\omega$, which is called as ``feasible region".
As shown in Fig. \ref{P-D}, we can find the feasible region of $\omega$ is $[0.41, 0.53]$ such that both the average delay constraint and average power constraint can be satisfied simultaneously.

\subsection{Effects of the Maximum Transmit Power}
We evaluate how the maximum transmit power $P_{\rm max}$ effects the delay performance and power consumption in the proposed algorithm.
Fig. \ref{P-D-PMAX} shows average power consumption and average delay with different maximum transmit powers, respectively.
As shown in Fig. \ref{P-D-PMAX}, the average power consumption increases with the maximum transmit power while the average delay decreases with the maximum transmit power.
This can be explained by the observation in Fig. \ref{P-PMAX}.
As the maximum transmit power gets larger, more power is consumed when the train moves at the cell edge, resulting that the buffered packets can be transmitted as soon as possible.

\begin{figure}[!htb]
  \centering
  \subfigure[]{
    \label{P-D-PMAX} 
    \includegraphics[width=0.93\linewidth]{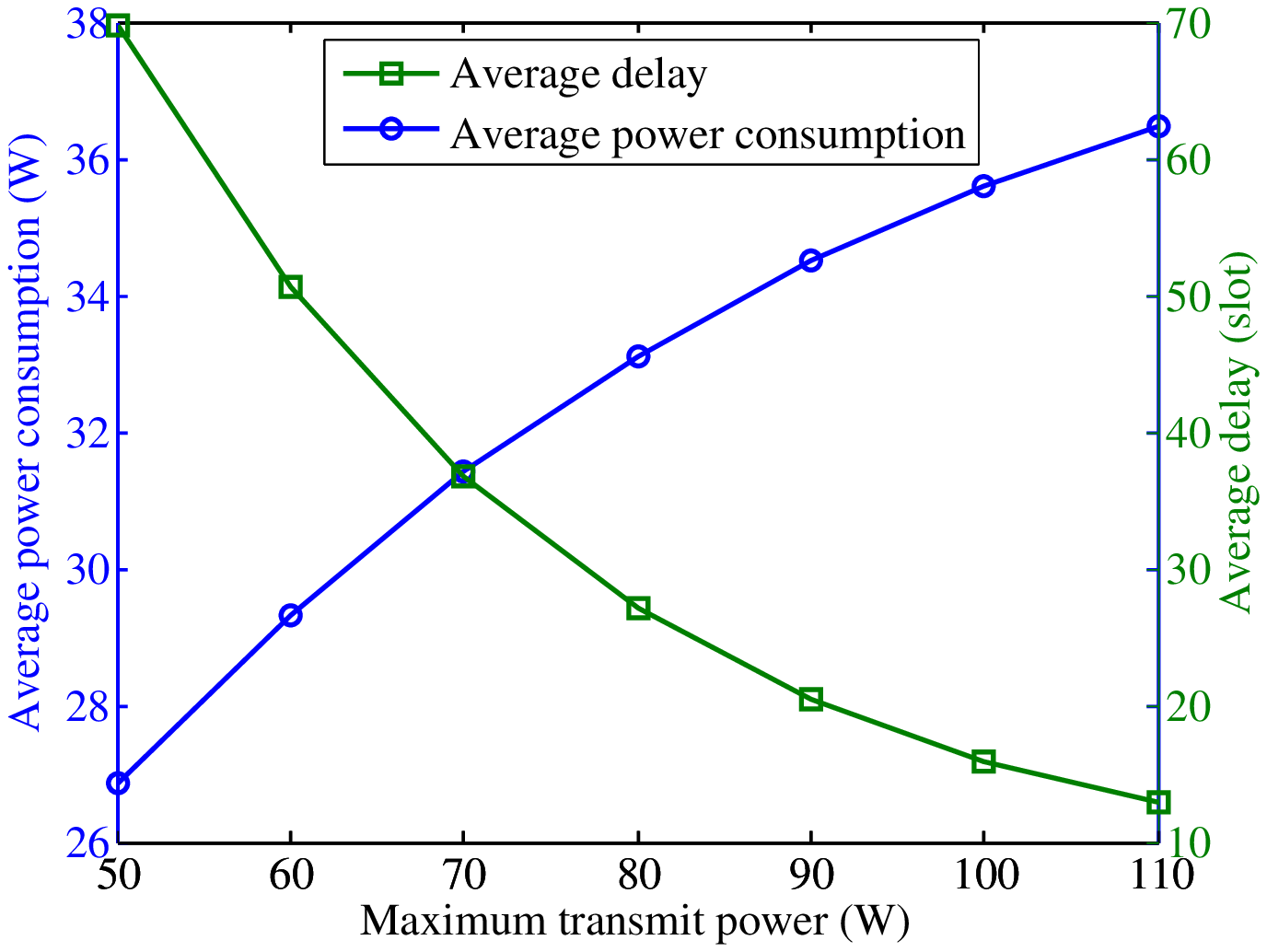}}
  \subfigure[]{
    \label{P-PMAX} 
    \includegraphics[width=0.88\linewidth]{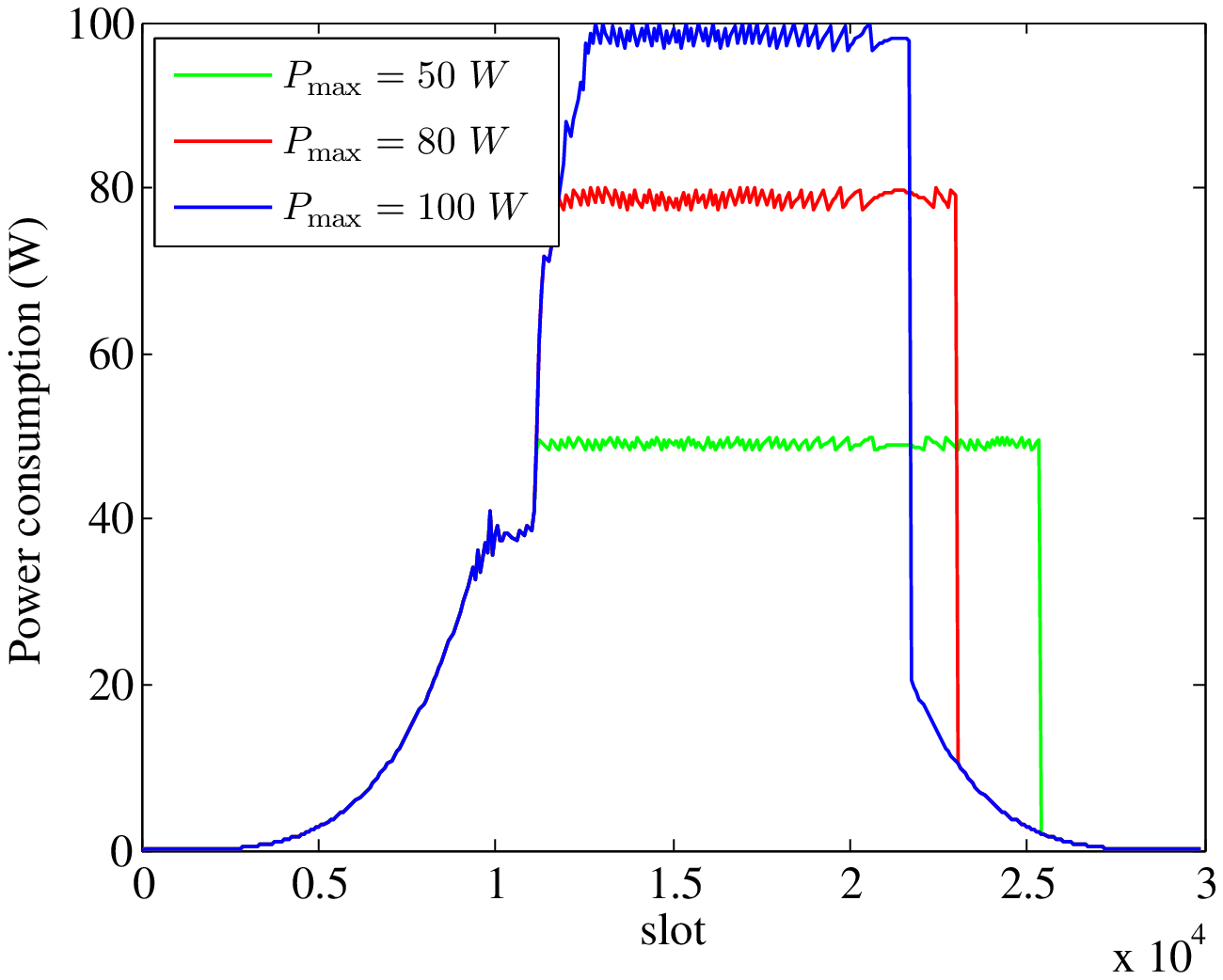}}
  \caption{Power consumption and delay performance under different maximum transmit powers, where $\lambda_{k} = 23$ packets/slot, $W_{k}^{{\rm av}} = 15$ slots and $\omega = 0.6$.}
  \label{Pathloss+Dopplershift} 
\end{figure}

\section{Conclusions}\label{sec:7}
In this paper, we investigate the delay-aware dynamic resource allocation and power control problem in HSR wireless communications.
The problem is formulated into a stochastic optimization problem, rather than pursuing the traditional convex optimization means.
A dynamic resource management algorithm is proposed to solve the intractable stochastic optimization problem.
The novelty of the proposed dynamic algorithm lies in the applications of stochastic network optimization approach and the ideas such as the virtual queue-based constraint transformation and opportunistically minimizing an expectation.
Simulation results are presented to show that the proposed dynamic algorithm can reasonably use the limited resource and significantly improve the delay performance under power constraints in HSR communications.

In the future, we plan to broaden and deepen this work in several directions.
First, we attempt to investigate delay-aware cross-layer design for multi-service transmission with more practical assumptions in HSR communication systems, e.g., multi-service resource management problem under more practical HSR channel model and more realistic packet arrival model.
Second, since the weight parameter plays a key role in balancing the queue delay and power consumption, we plan to analyze theoretically how to obtain the feasible region of the weight parameter.
Third, to address the issue that the delay requirement and power constraint can not be satisfied simultaneously, we would like to design a dynamic admission control scheme for supporting multi-service transmission.
Finally, as for the safety-related services transmission in HSR communications, we also plan to propose a new multi-service transmission scheme, which can take full consideration of different priorities and strict delay requirements.

\bibliographystyle{IEEEtran}
\bibliography{Ref_DSO}
\end{document}